\documentclass{article}

\usepackage{arxiv}

\newcommand{\sgn}{\operatornamewithlimits{sgn}}

\newcommand{\argmin}{\operatornamewithlimits{argmin}}

\newcommand{\sat}{\operatornamewithlimits{sat}}
\newcommand{\atan}{\operatornamewithlimits{atan2}}
\newcommand{\diagonal}{\operatornamewithlimits{diag}}

\newcommand{\dom}{\operatornamewithlimits{dom}}

\newtheoremstyle{preprintstyle}
  {10pt plus 2pt minus 2pt}
  {10pt plus 2pt minus 2pt}
  {\itshape}
  {}
  {\color{secondary}\bfseries}
  { - }
  {4pt}
  {}

\theoremstyle{preprintstyle}
\newtheorem{theorem}{\textbf{Theorem}}

\newtheorem{lemma}{\textbf{Lemma}}

\newtheorem{assumption}{\textbf{Assumption}}

\title{A Semi-global Hybrid Sensorless Observer for Permanent Magnets Synchronous Machines with Unknown Mechanical Model} 

\author{Alessandro~Bosso
		\thanks{ A. Bosso and A. Tilli are within the research group Advanced Control and Technologies for Enhanced Mechatronics and Automation (\href{https://dei.unibo.it/en/research/research-groups/actema}{\color{secondary} ACTEMA}).}
		 \And Ilario A. Azzollini
		 \thanks{I. A. Azzollini is with the Center for Research on Complex Automated Systems "Giuseppe Evangelisti" ({\color{secondary} CASY}).
		 \newline ACTEMA and CASY are at the Department of Electrical, Electronic and Information Engineering (\href{https://dei.unibo.it/en}{\color{secondary} DEI}), University of Bologna, Viale Risorgimento 2, 40136 Bologna, Italy.
		\newline Emails: [\href{mailto:alessandro.bosso3@unibo.it}{alessandro.bosso3},
		\href{mailto:ilario.azzollini@unibo.it}{ilario.azzollini},
		\href{mailto:andrea.tilli@unibo.it}{andrea.tilli}]@unibo.it
		}
		\And Andrea~Tilli \footnotemark[1]
		 }
	
\copyrightTransfer{\vskip -0.3in This work has been submitted to IFAC for possible publication}
\headerAuthor{A. Bosso, I. A. Azzollini and A. Tilli}

\begin{document}
\maketitle

\begin{abstract}                
In this paper, we present a hybrid sensorless observer for Permanent Magnets Synchronous Machines, with no a priori knowledge of the mechanical dynamics and without the typical assumption of constant or slowly-varying speed.
Instead, we impose the rotor speed to have a constant (unknown) sign and a non-zero magnitude at all times.
For the design of the proposed scheme, meaningful Lie group formalism is adopted to describe the rotor position as an element of the unit circle.
This choice, however, leads to a non-contractible state space, and therefore it introduces topological constraints that complicate the achievement of global/semi-global and robust results.
In this respect it is shown that the proposed observer, which employs a clock to periodically reset the estimates, is semi-globally practically asymptotically stable, and thus it improves a continuous-time version designed under the same assumptions.
As highlighted in the simulation results, the novel hybrid strategy leads to enhanced transient performance, notably without any modification of the gains employed in the continuous-time solution. 
These features motivate to augment the observer with a discrete-time identifier, leading to significantly faster rotor flux reconstruction.
\end{abstract}


\section{Introduction}

Permanent Magnets Synchronous Machines (PMSMs) are nowadays widely adopted in several fields, ranging from vehicle propulsion to industrial motion applications.
In many contexts, rotor position and speed information is required to achieve accurate regulation, yet the presence of mechanical sensors may pose reliability and economic issues.
Furthermore, such sensors often result impractical because of space and weight requirements, e.g. in small/medium size electrically-powered Unmanned Aerial Vehicles (UAVs).
In this respect, the so-called sensorless control techniques aim to replace the use of mechanical sensors with suitable reconstruction algorithms, and have been the subject of extensive research efforts.
Several design strategies from nonlinear control theory have been applied to sensorless control.
To name a few, without intending to be exhaustive, we recall Extended Kalman Filters \cite{hilairet2009speed}, Sliding Mode \cite{lee2013design} and High Gain and Adaptive strategies \cite{montanari2006speed, marino2008an, khalil2009speed}.
Recently, some works have been dedicated to stator resistance estimation \cite{verrelli2018persistency} and Interior Permanent Magnets Synchronous Machines (IPMSMs) \cite{ortega2019globally}.\\
In the field of sensorless control and observation, the capability of dealing with highly variable speed, with little to no a priori knowledge of the mechanical dynamics, becomes crucial to achieve high-end, high precision algorithms when motors are coupled with nonlinear time-varying loads.
This is the case, e.g., for the electric propulsion of UAVs or Hybrid Electric Vehicles (HEVs), where the environmental conditions heavily affect the external torque.
In \cite{bobtsov2015robust}, an observer performing voltage and current integration is used to reconstruct position and rotor flux amplitude, independently of the mechanical model.
In \cite{tilli2019towards}, a simple sixth-order observer with unknown mechanical model is designed employing a unit circle representation for the rotor angular configuration, and in \cite{bosso2020robust} such design is extended to include resistance estimation, combined with appropriate signal injection techniques to ensure observability.
Indeed, the unit circle (indicated with $\mathbb{S}^1$) is a compact abelian Lie group, and Lyapunov-based tools can be used to derive a simple stability analysis.
Specifically, the algorithm in \cite{tilli2019towards} exploits a high gain observer to reconstruct the back-Electromotive Force (back-EMF), which is then used to set up an adaptive attitude estimator on $\mathbb{S}^1$.
The resulting reduced-order dynamics, corresponding to the attitude observer reconstruction error, is shown to evolve on the cylinder $\mathbb{S}^1\times\mathbb{R}$.\\
The use of a compact Lie group representation, however, introduces some relevant challenges.
In fact, it is known that when a dynamical system evolves on a manifold that is not diffeomorphic to any Euclidean space, it is impossible for a continuous vector field to globally asymptotically stabilize an equilibrium point \cite{mayhew2011quaternion}.
This phenomenon clearly arises in \cite{tilli2019towards} since two isolated hyperbolic equilibria are present, a stable node/focus and a saddle point: this restricts the basin of attraction of the reduced-order dynamics to $(\mathbb{S}^1\times\mathbb{R})\backslash\mathcal{R}_\mathcal{U}$, where $\mathcal{R}_\mathcal{U}$ is a curve passing through the saddle equilibrium.
This property directly affects the full-order observer and therefore only regional stability can be established.
We report the attitude observers on $SO(3)$ studied in \cite{mahony2008nonlinear}, which display a similar behavior in a higher dimensional context.
Notably, an attempt to break this kind of topological constraints with discontinous and memoryless feedback laws leads to non-robust solutions, causing in practice chattering behaviors.
Indeed, it is known that a dynamic hybrid feedback law must be employed in order to achieve global \textit{and} robust results \cite{sontag1999clocks, mayhew2011quaternion}.\\
In this work, we introduce a hybrid modification to the position, speed and flux observer of \cite{tilli2019towards} with the aim of establishing semi-global instead of regional stability.
An alternative to the continuous-time strategy is possible because both components of the back-EMF vector are available as indirect measurement, thus allowing to detect when the angular estimation error settles to a wrong configuration.
Exploiting this fact, we introduce a simple strategy based on a clock to periodically reset the position reconstruction.
The rotor speed is restricted to have a constant (unknown) sign and to be bounded in norm from above and below by positive scalars.
These conditions are compatible with many applications, including the control of propeller motors, where the sign of speed is usually not reversed.
The properties of the observer are highlighted by means of two time scales arguments (see e.g. \cite{teel2003unified, sanfelice2011singular}), and numerically compared to the structure in \cite{tilli2019towards}.
In particular, we underline the interesting feature that if the same tuning gains are adopted for the observer flows, the new algorithm displays a consistently faster transient response.
Finally, inspired by these enhanced convergence properties, we also propose an augmentation based on a discrete-time identifier to further boost the estimation performance, clearly at the expense of increased computational complexity.\\
The structure of the paper is the following.
After a brief introduction to the mathematical background in Section \ref{sec:notation}, we formally state the observer problem in Section \ref{sec:problem}.
The observer structure and its stability properties are presented in Section \ref{sec:main_result}, while in Section \ref{sec:conclusions} some concluding remarks and future research directions are outlined.

\section{Notation} \label{sec:notation}

\par We use $(\cdot)^T$ to denote the transpose of real-valued matrices.
For compactness of notation we often indicate with $(v, w)$, for any pair of column vectors $v$, $w$, the concatenated vector $(v^T, w^T)^T$.
In case of non-differentiable signals, the upper right Dini derivative, indicated with $D^+$, is employed as generalized derivative. The time argument of signals will be omitted when clear from the context.

\subsection{Notation on the Unit Circle ($\mathbb{S}^1$)}

We employ the unit circle $\mathbb{S}^1$ to represent reference frames involved in the manipulation of PMSM equations, as in \cite{tilli2019towards}. In particular, $\mathbb{S}^1$ is a compact abelian Lie group, with the planar 2-D rotation employed as group operation. An integrator on $\mathbb{S}^1$ is given by
\begin{equation*}
	\dot{\zeta} = u(t)
	\underbrace{
	\begin{pmatrix}
	0 & -1 \\
	1 & 0
	\end{pmatrix}}_{\mathcal{J}} \zeta  ,
	\qquad
	\zeta \in \mathbb{S}^1,
\end{equation*}
with $u(t) \in \mathbb{R}$.
Any angle $\vartheta \in \mathbb{R}$ can be mapped into an element of the unit circle given by $(\cos(\vartheta) \; \sin(\vartheta))^T \in \mathbb{S}^1$. The identity element in $\mathbb{S}^1$ is $(1 \; 0)^T$. Finally, to any $\zeta = (c\; s)^T \in \mathbb{S}^1$ we can associate a rotation matrix $\mathcal{C}[\zeta] = {\tiny\begin{pmatrix}c & -s\\ s & c\end{pmatrix}}$, which is used for group multiplication: for any $\zeta_1, \zeta_2 \in \mathbb{S}^1$, the product is given by $\zeta_1 \cdot \zeta_2 = \mathcal{C}[\zeta_1]\zeta_2 = \mathcal{C}[\zeta_2]\zeta_1$.

\subsection{Hybrid Dynamical Systems}

In this paper we adopt the formalism of hybrid dynamical systems as in \cite{goebel2012hybrid}. In particular, a hybrid system $\mathcal{H}$ can be described as
\begin{equation}\nonumber
\mathcal{H} : \left\{
\begin{split}
&\dot{x} &\in F(x,u) & \qquad (x, u) \in C\\
&x^+ &\in G(x,u) & \qquad (x, u) \in D
\end{split}
\right.
\end{equation}
where $x$ is the state, $u$ is the input, $C$ is the flow set, $F$ is the flow map, $D$ is the jump set, and $G$ is the jump map. The state of the hybrid system can either flow according to the differential inclusion $\dot{x} \in F$ (while $(x, u) \in C$), or jump according to the difference inclusion ${x}^+ \in G$ (while $(x, u) \in D$). For all the concepts regarding hybrid solutions, stability, robustness, and related Lyapunov theory, we refer to \cite{goebel2012hybrid} and references therein.

\section{Model Formulation and Problem Statement}\label{sec:problem}

The electromagnetic model of a PMSM in a static bi-phase reference frame under balanced working conditions, linear magnetic circuits, and negligible iron losses, can be written as
\begin{equation}\label{eq:model1}
\frac{d}{dt}i_s = -\frac{R}{L}i_s + \frac{1}{L}u_s - \frac{\omega \varphi \mathcal{J} \zeta}{L}, \qquad \dot{\zeta} = \omega \mathcal{J} \zeta, 
\end{equation}
where $i_s, u_s \in \mathbb{R}^2$ are the stator currents and voltages, respectively, and in particular $u_s$ is a piecewise continuous signal defined on the interval $[t_0, \infty)$, with $t_0$ the initial time.
Furthermore, $\omega$ is the rotor electrical angular speed, while $\zeta \in \mathbb{S}^1$ and $\varphi \in \mathbb{R}_{>0}$ are the angular configuration and the (constant) amplitude of the rotor magnetic flux vector, respectively.
Finally, $R$ is the stator resistance and $L$ is the stator inductance.\\
In the field of sinusoidal machines, it is common practice to represent \eqref{eq:model1} in rotating reference frames. Consider a generic rotating reference frame with $\zeta_r \in \mathbb{S}^1$ and $\omega_r$ its angular orientation and speed, respectively. Then, \eqref{eq:model1} becomes
\begin{equation}\label{eq:model2}
\begin{split}
\frac{d}{dt}i_r &= -\frac{R}{L}i_r + \frac{1}{L}u_r - \frac{\omega \varphi \mathcal{J} \mathcal{C}^T[\zeta_r] \zeta}{L} - \omega_r \mathcal{J} i_r \\
\dot{\zeta} &= \omega \mathcal{J} \zeta, \qquad \dot{\zeta}_r = \omega_r \mathcal{J} \zeta_r,
\end{split}
\end{equation}
where $i_r = \mathcal{C}^T[\zeta_r]i_s$, $u_r = \mathcal{C}^T[\zeta_r]u_s$.\\
In this work, the angular speed $\omega$ is modeled as an unknown bounded input, and the following regularity assumption is required.
\begin{assumption}\label{assumption}\label{hyp:w_regularity}
The signal $\omega(\cdot)$ is defined over the interval $[t_0,\infty)$ and, in addition:
\begin{itemize}
	\item $\omega(\cdot)$ is $\mathcal{C}^0$ and piecewise $\mathcal{C}^1$ in its domain of existence;
	\item there exist positive scalars $\omega_{\min}$, $\omega_{\max}$ such that, for all $t \geq t_0$, it holds $\omega_{\min} \leq |\omega(t)| \leq \omega_{\max}$;
	\item $|D^+ \omega(t)|$ exists and is bounded, for all $t \geq t_0$.
\end{itemize}
\end{assumption}
Note that these conditions, combined with $u_s$ being piecewise continuous, ensure existence and uniqueness of solutions on $[t_0, \infty)$.
Additionally, since the properties that we specified for the input signals do not depend on $t_0$, we can choose in the following $t_0 = 0$ without loss of generality.
Assumption \ref{assumption} requires the angular speed $\omega$ to have constant sign and uniformly non-zero magnitude.
This condition is slightly more restrictive than the well-known assumption of non-permanent zero speed, which was proven to be a sufficient condition to reconstruct $\omega$, $\zeta$, and $\varphi$, assuming currents and voltages available for measurement and the parameters $R$ and $L$ perfectly known \cite{zaltni2010synchronous}. Nevertheless, Assumption \ref{assumption} is compatible with significant applications such as renewables electric energy generation and electric vehicles propulsion (UAVs, HEVs).\\
We finally recall the problem of \textit{sensorless observer, with (restricted) variable speed and no mechanical model} \cite{tilli2019towards}: given the PMSM dynamics \eqref{eq:model1} or \eqref{eq:model2}, design an estimator of $\zeta$, $\omega$, $\varphi$ with only stator voltages and currents available for measurement, such that appropriate stability and convergence properties hold under Assumption \ref{assumption}.

\section{The Proposed Hybrid Observer}\label{sec:main_result}

\begin{table}[t!]	
	\begin{center}
		\captionsetup{width=0.7\columnwidth}
		\caption{System and observers parameters}\label{tab:TabParMot}
		\begin{tabular}{lc | lc}\hline
			\hline
			{\scriptsize Stator resistance $R$ $[\Omega]$}&{\scriptsize $0.06$}&{\scriptsize $k_{\text{p}}$}&{\scriptsize $2.18 \times 10^4$}  \\
			{\scriptsize Stator inductance $L$ [$\mu \text{H}$]}&{\scriptsize $33.75$} &{\scriptsize $k_{\text{i}}$}&{\scriptsize $9.34 \times 10^3$}\\
			{\scriptsize Nominal angular speed [$\text{rpm}$]}&{\scriptsize $6000$} & {\scriptsize $k_{\eta}$}&{\scriptsize $95.7$} \\
			{\scriptsize Rotor magnetic flux $\varphi$ [$\text{mWb}$]}&{\scriptsize $1.9$}&{\scriptsize $\gamma$}&{\scriptsize $4582$}\\
			{\scriptsize Number of pole pairs $p$}&{\scriptsize $7$}&{\scriptsize $\Lambda$}&{\scriptsize $200$}\\
			{\scriptsize Load Inertia [$\text{Kgm}^2$]}&{\scriptsize $2.5 \times 10^{-5}$}&{\scriptsize $N$}&{\scriptsize $2$}\\
			\hline
			\end{tabular}
		\end{center}
		
		\vspace{-8pt}
		
\end{table}

In this section we present the main result of this work, and we compare it with a preliminary continuous-time solution.
To simplify the presentation and better highlight the connection between the different strategies, we choose to embed along the text some numerical results, based on a UAV propeller motor, whose parameters are indicated in Table \ref{tab:TabParMot}.
Since the observer transient performance is more evident when it is disconnected from the controller, we employ a standard sensorized field-oriented controller to generate the speed profile $\omega$, selected as a combination of constant and time-varying ``aggressive'' sequences.
We omit the complete closed-loop simulations for brevity.

\subsection{The $\chi$-Reference Frame and a Continuous-Time Solution}
Let $\chi \coloneqq |\omega|\varphi \in \mathbb{R}_{>0}$, $\xi \coloneqq (1/\varphi)\sgn(\omega)$.
This allows to consider, in order to replace the PMSM angular dynamics, the particular frame given by $\zeta_\chi \coloneqq \zeta \sgn(\xi) = \zeta \sgn(\omega)$, which yields a very simple reformulation of the model, for a generic rotating frame $\zeta_r$:
\begin{equation}\label{eq:chi_frame}
\begin{split}
\frac{d}{dt}i_r &= -\frac{R}{L}i_r + \frac{1}{L}u_r - \frac{\chi\mathcal{J}\mathcal{C}^T[\zeta_r]\zeta_\chi}{L} - \hat{\omega}_r\mathcal{J}i_r\\
\dot{\zeta}_\chi &= \chi\xi \mathcal{J}\zeta_\chi, \qquad \dot{\zeta}_r = \omega_r \mathcal{J} \zeta_r. 
\end{split}
\end{equation}
Note, in particular, that $\xi \in \mathbb{R}$ is an unknown parameter and $\chi$ satisfies the following properties, which follow as a direct consequence of Assumption \ref{assumption}:
\begin{itemize}
\item $\chi$ is $\mathcal{C}^0$ and piecewise $\mathcal{C}^1$;
\item $\chi_{\text{m}} \leq \chi \leq \chi_{\text{M}}$, for some positive scalars $\chi_{\text{m}}$, $\chi_{\text{M}}$;
\item $|D^+\chi| \leq M$, for some positive scalar $M$.
\end{itemize}
The main idea of the proposed observer is then to design an estimator of the frame $\zeta_\chi$ by using as representation a frame $\zeta_r = \hat{\zeta}_\chi$, whose dynamics needs to be designed appropriately so that the two references synchronize asymptotically.
This synchronization problem can be recast as the stabilization of the misalignment error $\eta \coloneqq \mathcal{C}^T[\hat{\zeta}_\chi]\zeta_\chi \in \mathbb{S}^1$.
For convenience, we use the subscript $(\cdot)_{\hat{\chi}}$ to indicate the electric variables in the frame $\hat{\zeta}_\chi$, leading to the following dynamics:
\begin{equation}
\begin{split}
\frac{d}{dt}i_{\hat{\chi}} &= -\frac{R}{L}i_{\hat{\chi}} + \frac{1}{L}u_{\hat{\chi}} - \frac{\chi\mathcal{J}\eta}{L} - \hat{\omega}_\chi\mathcal{J}i_{\hat{\chi}}\\
\dot{\zeta}_{\chi} &= \chi \xi \mathcal{J}\zeta_\chi, \qquad \dot{\hat{\zeta}}_\chi = \hat{\omega}_\chi \mathcal{J}\hat{\zeta}_\chi,
\end{split}
\end{equation}
with $\hat{\omega}_\chi$ the angular speed of the frame $\hat{\zeta}_\chi$.
In \cite{tilli2019towards}, the synchronization problem was addressed with a continuous-time observer of the form
\begin{equation} \label{eq:observer}
\begin{split}
\dot{\hat{\imath}} &= -\frac{R}{L}\hat{\imath} + \frac{1}{L}u_{\hat{\chi}} + \frac{\hat{h}}{L} - \left(|\hat{h}|\hat{\xi} + k_{\eta}\hat{h}_1\right)\mathcal{J}i_{\hat{\chi}} + k_{\text{p}}\tilde{\imath}\\
\dot{\hat{h}} &= k_{\text{i}}\tilde{\imath} \qquad \dot{\hat{\zeta}}_\chi = \left(|\hat{h}|\hat{\xi} + k_{\eta}\hat{h}_1\right)\mathcal{J}\hat{\zeta}_\chi \qquad \dot{\hat{\xi}} = \gamma \hat{h}_1
\end{split}
\end{equation}
where $\hat \imath \in \mathbb{R}^2$ is the reconstruction of $i_{\hat{\chi}}$, $\tilde{\imath} \coloneqq  i_{\hat{\chi}} - \hat{\imath}$ is the current estimation error, $\hat \xi \in \mathbb{R}$ is the estimation of $\xi$, $\hat{h}$ is the estimate of the back-EMF $h = -\chi\mathcal{J}\eta$, while the angular speed is assigned as $\hat{\omega}_\chi = |\hat{h}|\hat{\xi} + k_{\eta}\hat{h}_1$.
Finally, $k_{\text{p}}$, $k_{\text{i}}$, $k_\eta$ and $\gamma$ are positive scalars used for tuning.\\
With this structure, the outputs of the observer, providing the desired estimates, are given by $\hat{\omega} = |\hat{h}|\hat{\xi}$ (the term multiplying $\hat{h}_1$ is often omitted in practice to limit noise propagation), $\hat{\zeta} = \hat{\zeta}_\chi\sgn(\hat{\xi})$ and $\hat{\varphi} = \sat(1/{|\hat{\xi}|})$, where the bounds of the saturation are chosen according to the expected motor parameter ranges.
In particular, it was proven that as long as Assumption \ref{assumption} holds, it is possible to ensure regional practical asymptotic stability by proper selection of the gains.
This stems from the two time scales approach used to separate the dynamics into a fast subsystem, given by a high-gain observer for current and back-EMF estimation, and a slow subsystem, which can be interpreted as an adaptive attitude observer on $\mathbb{S}^1$.
From the state space of the slow subsystem derives the inherently regional and not semi-global result: the reduced order error dynamics, obtained by supposing perfect knowledge of $i_{\hat{\chi}}$ and $h$, can be written as follows on the cylinder $\mathbb{S}^1 \times \mathbb{R}$, with $\tilde{\xi} \coloneqq \xi - \hat{\xi}$ indicating the flux estimation error:
\begin{equation}\label{eq:reduced}
\dot{\eta} = \left( \chi\tilde{\xi} - k_{\eta} \chi \eta_2 \right)\mathcal{J}\eta \qquad \dot{\tilde{\xi}} = -\gamma\chi \eta_2.
\end{equation}
Indeed, it can be proven that the domain of attraction of the configuration $\bar{x}_{\text{s}} = ((1, 0), 0) \in \mathbb{S}^1\times\mathbb{R}$, corresponding to rotation alignment and correct flux estimation, does not include an unstable manifold of dimension $1$, which originates from the saddle equilibrium $\bar{x}_{\text{u}} = ((-1, 0), 0)$ as shown in red in Figure \ref{fig:phase_diagram}.
\begin{figure}[t]
	\centering
	\psfragscanon
	\psfrag{x} [B][B][0.7][0]{$\tilde{\vartheta}$}
	\psfrag{y} [B][B][0.7][0]{$\tilde{\xi}$}
	\psfrag{a} [B][B][0.7][0]{$\bar{x}_{\text{u}}$}
	\psfrag{b} [B][B][0.7][0]{$\bar{x}_{\text{s}}$}
	\psfrag{c} [B][B][0.7][0]{$\bar{x}_{\text{u}}$}
	\psfrag{-p} [B][B][0.7][0]{$-\pi$}
	\psfrag{p} [B][B][0.7][0]{$\pi$}
	\includegraphics[clip=true, height = 4cm]{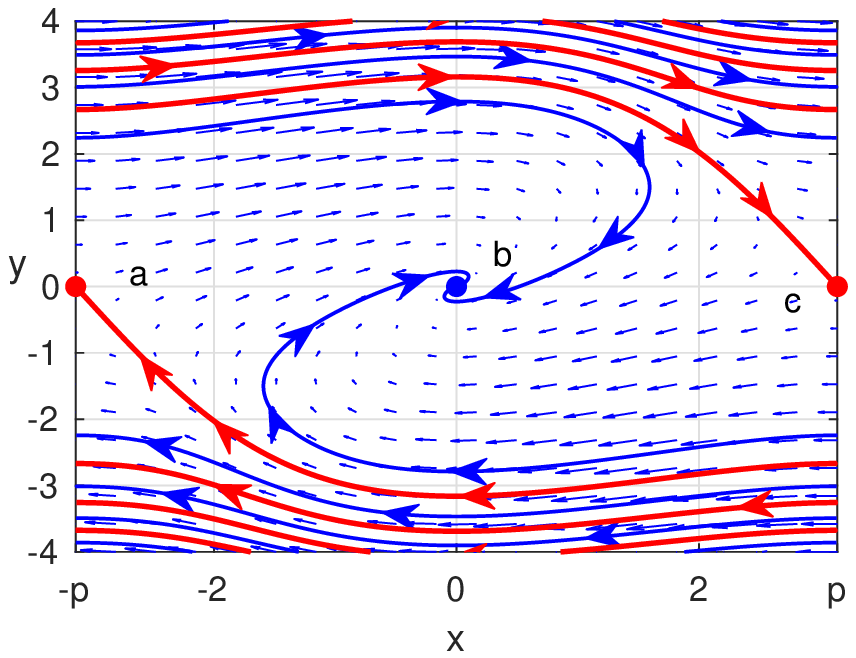}
	\psfrag{x} [B][B][0.7][0]{$\eta_1$}
	\psfrag{y} [B][B][0.7][0]{$\eta_2$}
	\psfrag{z} [B][B][0.7][0]{$\tilde{\xi}$}
	\includegraphics[clip=true, height=4cm]{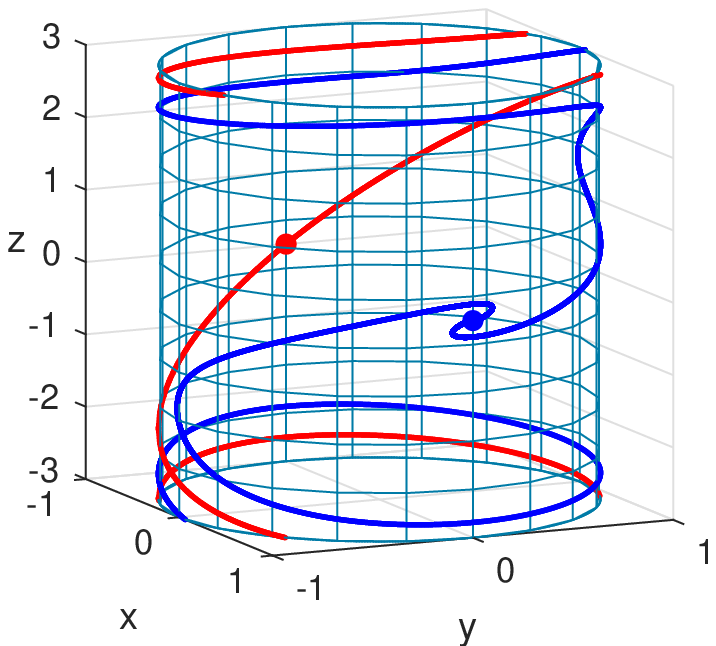}
	\caption{Phase diagram of \eqref{eq:reduced} for $\chi = 1$, $k_\eta = 1.5$, $\gamma = 1$, as shown in \cite{tilli2019towards}. On the left, the unstable manifold (red) and some trajectories converging to $\bar{x}_{\text{s}}$ (blue) are depicted on the equivalent planar representation $(\tilde{\vartheta}, \tilde{\xi})$, where $\tilde{\vartheta} = \atan(\eta_2, \eta_1)$ is the unique angle in the interval $[-\pi; \pi)$ corresponding to $\eta$. On the right, the same objects are represented on the cylinder $\mathbb{S}^1 \times \mathbb{R}$.}
	\label{fig:phase_diagram}
	
	\vspace{-8pt}
	
\end{figure}
Figure \ref{fig:sim_results} (plots (a),(d),(g),(j)) presents the simulation results corresponding to observer \eqref{eq:observer} with the same gains as in \cite{tilli2019towards} (see Table \ref{tab:TabParMot}).
Note that the initial transient (corresponding to high values of $\tilde{\xi}$) is relatively slow, highlighting the same helicoidal shape as in Figure \ref{fig:phase_diagram}.

\begin{figure}[t!]
 	\centering
	\begin{subfigure}[b]{0.23\textwidth}
	\centering
 	\psfrag{x}[B][B][0.7][0]{(a) \; time [s]}
 	\psfrag{y}[B][B][0.7][0]{$\omega/p$, $\hat{\omega}/p$ [$\text{rpm}$]}
 	\raisebox{-\height}{\includegraphics[clip = true, width = \textwidth]{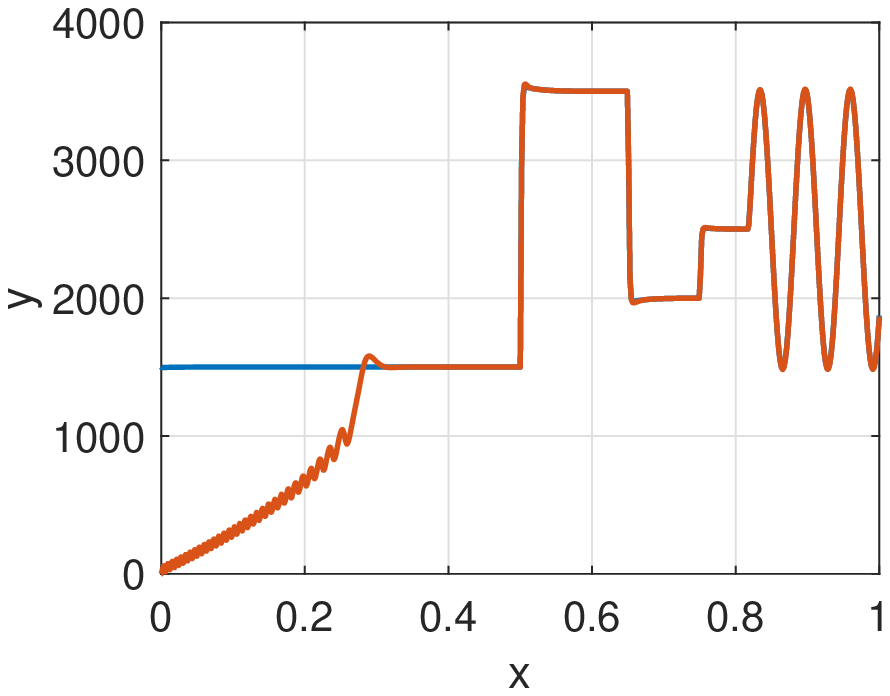}}
	
	\vspace{2pt}
	
	\centering
 	\psfrag{x}[B][B][0.7][0]{(b) \; time [s]}
 	\psfrag{y}[B][B][0.7][0]{$\omega/p$, $\hat{\omega}/p$ [$\text{rpm}$]}
 	\raisebox{-\height}{\includegraphics[clip = true, width = \textwidth]{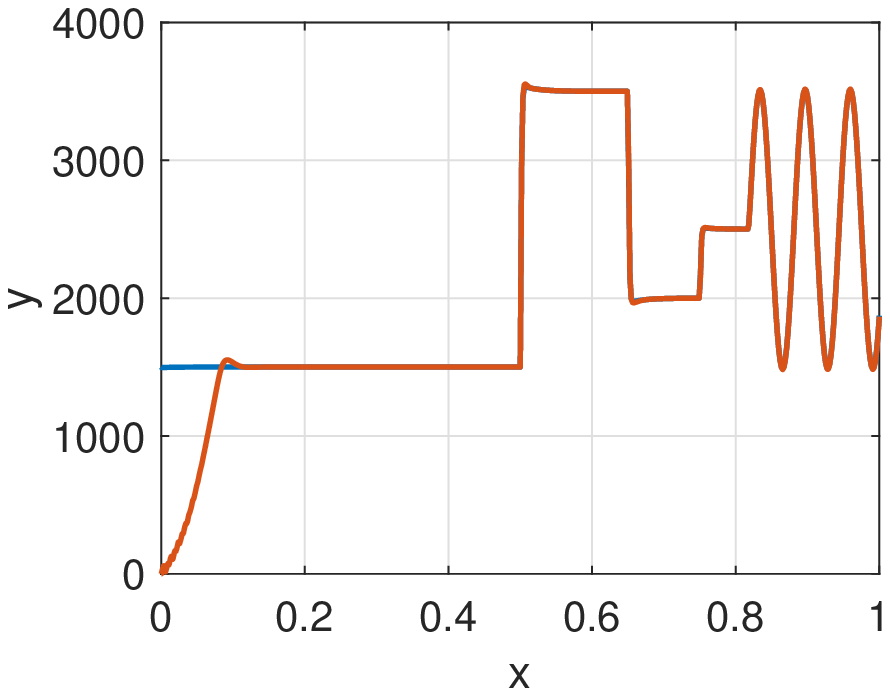}}
	
	\vspace{2pt}
	
 	\psfrag{x}[B][B][0.7][0]{(c) \; time [s]}
 	\psfrag{y}[B][B][0.7][0]{$\omega/p$, $\hat{\omega}/p$ [$\text{rpm}$]}
 	\raisebox{-\height}{\includegraphics[clip = true, width = \textwidth]{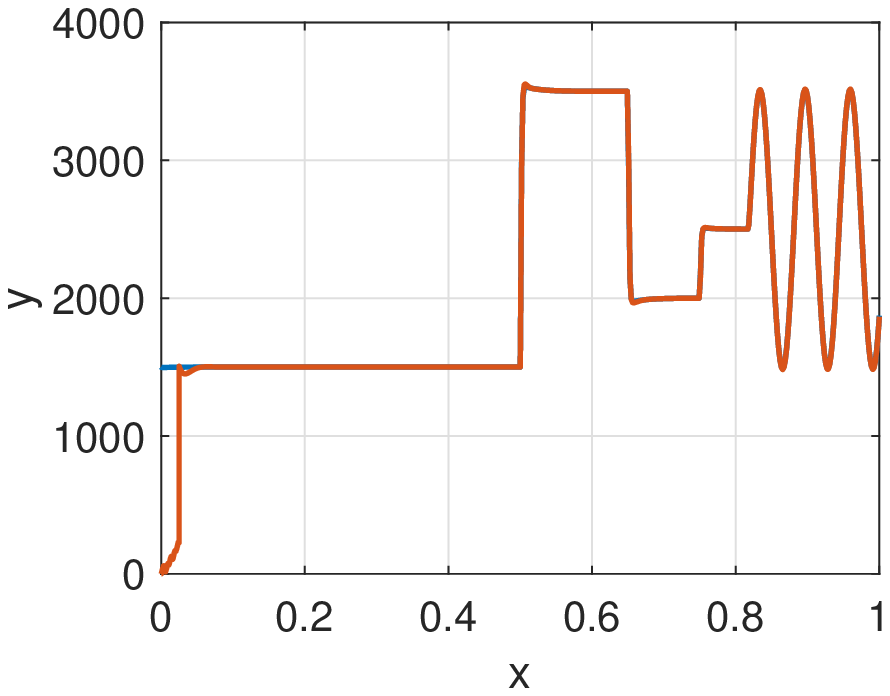}}
	\end{subfigure}
	\hspace{0.01\textwidth}
	\begin{subfigure}[b]{0.23\textwidth}	
	\centering
 	\psfrag{x}[B][B][0.7][0]{(d) \; time [s]}
 	\psfrag{y}[B][B][0.7][0]{$\tilde{\vartheta}$ [$\text{rad}$]}
 	\raisebox{-\height}{\includegraphics[clip = true, width = \textwidth]{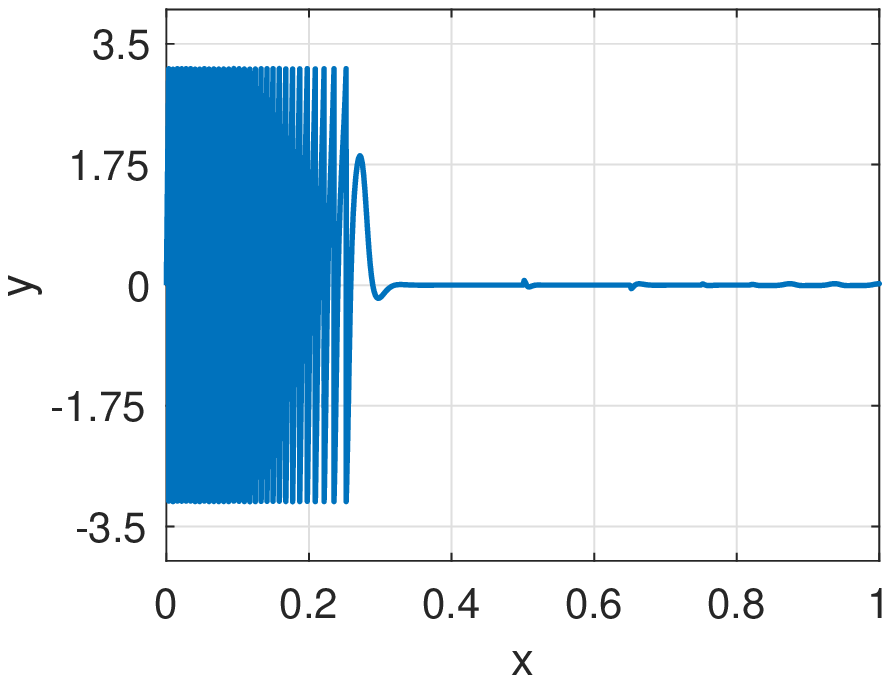}}
	
	\vspace{2pt}
		
	\centering
 	\psfrag{x}[B][B][0.7][0]{(e) \; time [s]}
 	\psfrag{y}[B][B][0.7][0]{$\tilde{\vartheta}$ [$\text{rad}$]}
 	\raisebox{-\height}{\includegraphics[clip = true, width = \textwidth]{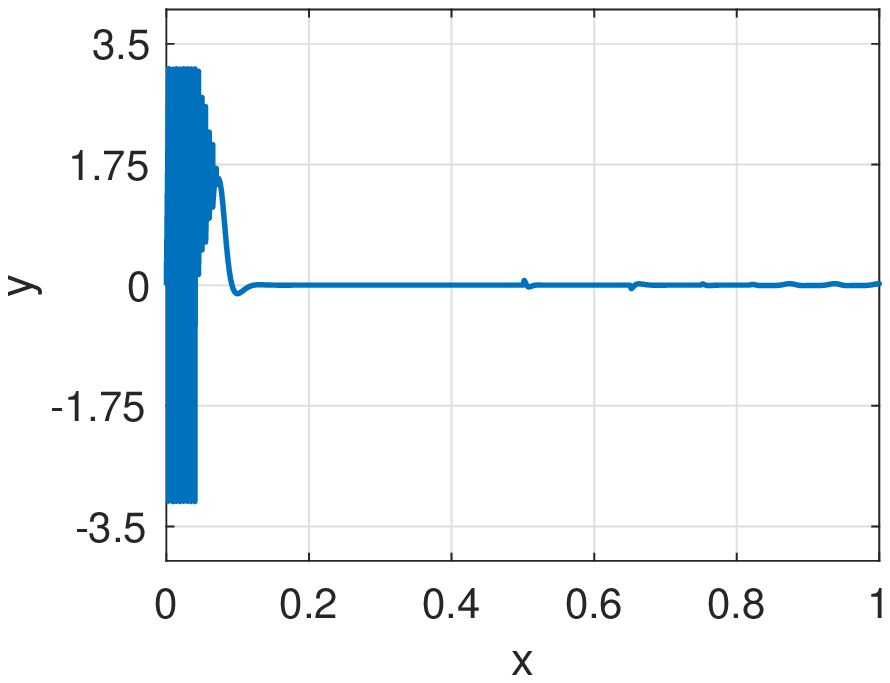}}
	
	\vspace{2pt}
 	
	\psfrag{x}[B][B][0.7][0]{(f) \; time [s]}
 	\psfrag{y}[B][B][0.7][0]{$\tilde{\vartheta}$ [$\text{rad}$]}
 	\raisebox{-\height}{\includegraphics[clip = true, width = \textwidth]{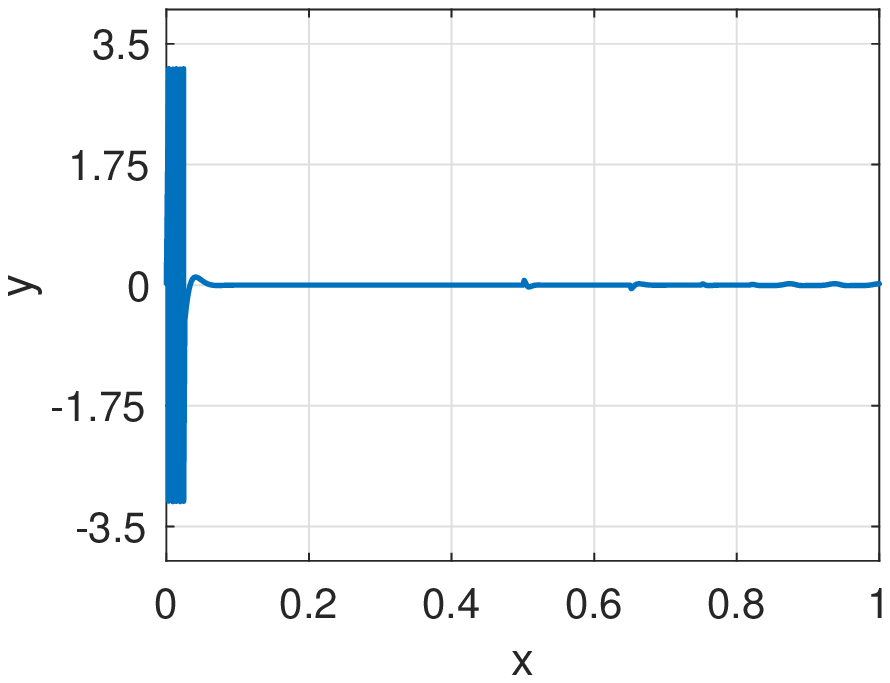}}
	\end{subfigure}
	\hspace{0.01\textwidth}
	\begin{subfigure}[b]{0.23\textwidth}
	\centering
 	\psfrag{x}[B][B][0.7][0]{(g) \; time [s]}
 	\psfrag{y}[B][B][0.7][0]{$\xi$, $\hat{\xi}$ [${\text{Wb}}^{-1}$]}
 	\raisebox{-\height}{\includegraphics[clip = true, width = \textwidth]{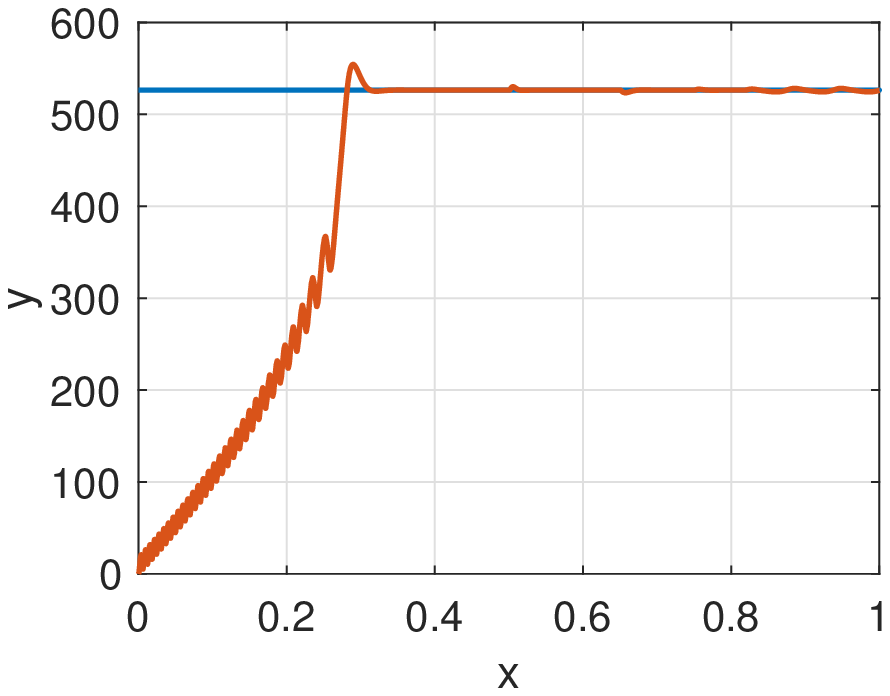}}
	
	\vspace{2pt}
	
	\centering
 	\psfrag{x}[B][B][0.7][0]{(h) \; time [s]}
 	\psfrag{y}[B][B][0.7][0]{$\xi$, $\hat{\xi}$ [${\text{Wb}}^{-1}$]}
 	\raisebox{-\height}{\includegraphics[clip = true, width = \textwidth]{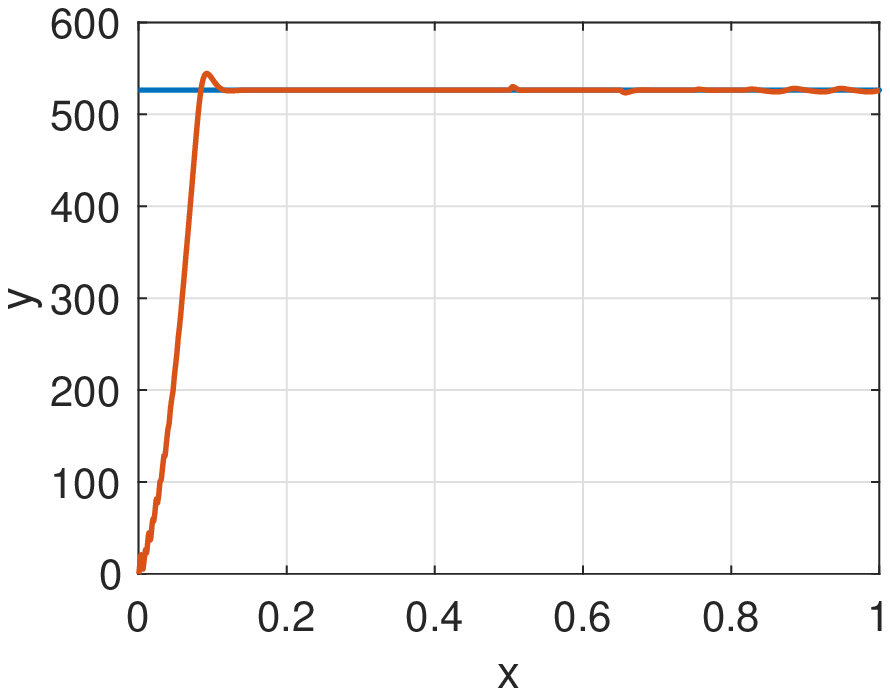}}
	
	\vspace{2pt}
	
	\psfrag{x}[B][B][0.7][0]{(i) \; time [s]}
 	\psfrag{y}[B][B][0.7][0]{$\xi$, $\hat{\xi}$ [${\text{Wb}}^{-1}$]}
 	\raisebox{-\height}{\includegraphics[clip = true, width = \textwidth]{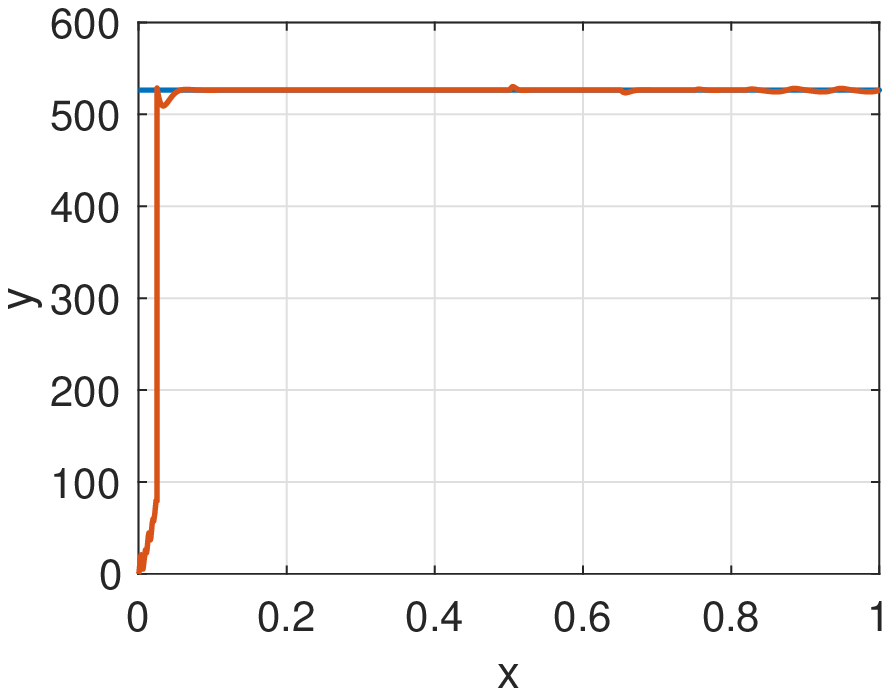}}
	\end{subfigure}
	\hspace{0.01\textwidth}
	\begin{subfigure}[b]{0.23\textwidth}
	\centering
 	\psfrag{x}[B][B][0.7][0]{(j) \; time [s]}
 	\psfrag{y}[B][B][0.7][0]{$\tilde{h}$ [$\text{V}$]}
 	\raisebox{-\height}{\includegraphics[clip = true, width = \textwidth]{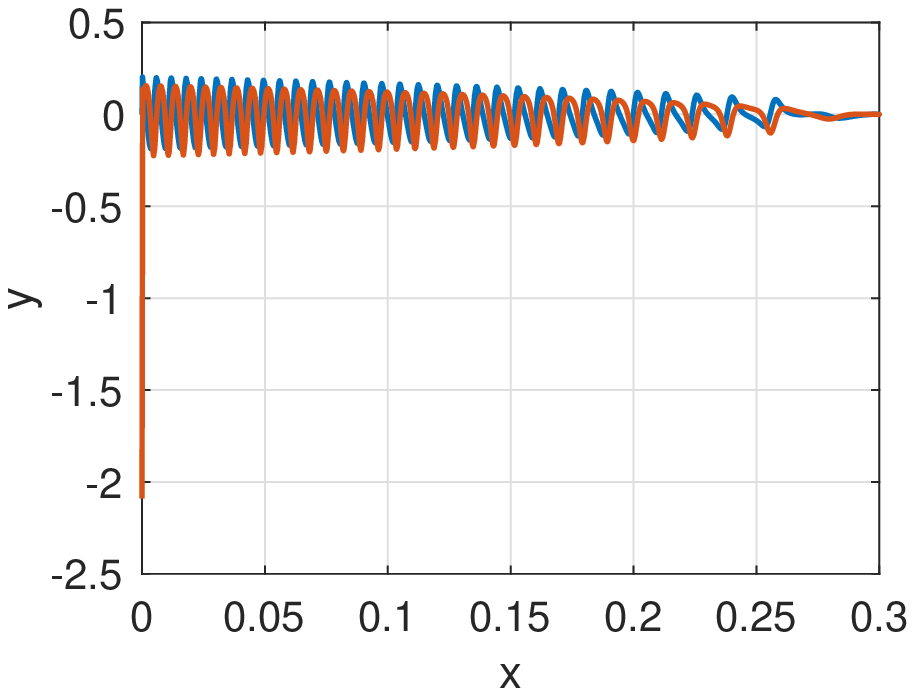}}
	
	\vspace{2pt}
	
	\centering
 	\psfrag{x}[B][B][0.7][0]{(k) \; time [s]}
 	\psfrag{y}[B][B][0.7][0]{$\tilde{h}$ [$\text{V}$]}
 	\raisebox{-\height}{\includegraphics[clip = true, width = \textwidth]{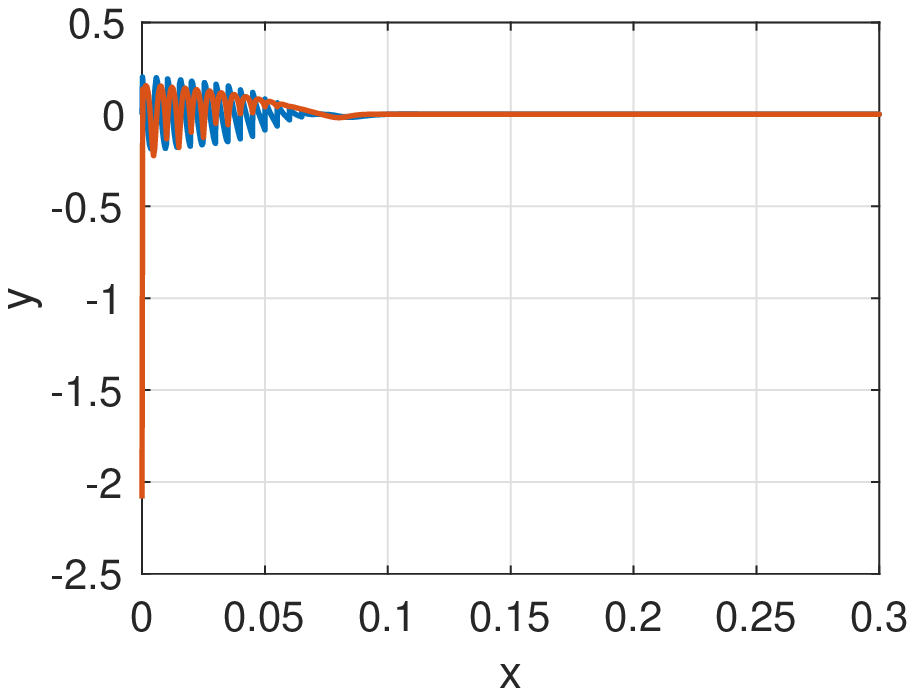}}
	
	\vspace{2pt}
	
	\psfrag{x}[B][B][0.7][0]{(l) \; time [s]}
 	\psfrag{y}[B][B][0.7][0]{$\tilde{h}$ [$\text{V}$]}
 	\raisebox{-\height}{\includegraphics[clip = true, width = \textwidth]{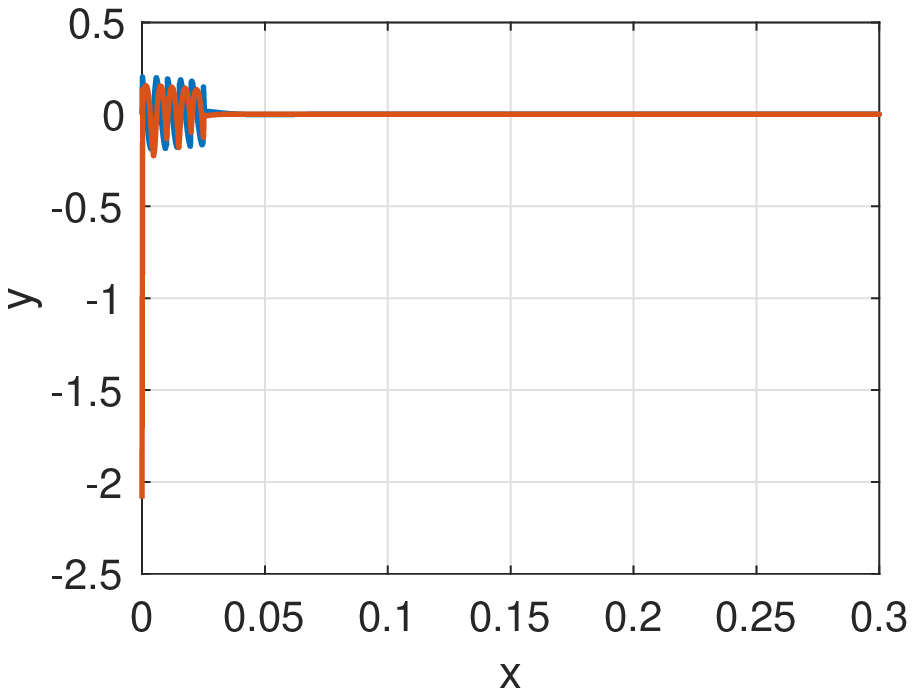}}
	\end{subfigure}
 	\caption{First row: observer \eqref{eq:observer}. Second row: observer \eqref{eq:hybrid_obs_v1}. Third row: observer \eqref{eq:hybrid_obs_v1}-\eqref{eq:batch}-\eqref{eq:batch_jump}. (a),(b),(c): Rotor angular speed (blue) and estimated value (red). (d),(e),(f): Rotor angular position reconstruction error. (g),(h),(i): Parameter $\xi$ (blue) and its estimate (red). (j),(k),(l): Back-EMF reconstruction error, with the first component in blue and the second one in red.}
 	\label{fig:sim_results}
	
 	\vspace{-8pt}
	
\end{figure} 

\subsection{A Hybrid Strategy for Semi-Global Stability}

Following the insights provided by the continuous-time solution, we opt to modify the reduced order system \eqref{eq:reduced} by enriching its dynamics with a jump policy, which corresponds to jumps of the estimates $\hat{\zeta}_\chi$, $\hat{\xi}$, while preserving the existent flows.
To simplify the approach and allow easy implementation of the observer, we propose to augment the observer dynamics with a clock, given by:
\begin{equation}\label{eq:timer}
\left\{
\begin{matrix}
\dot{\rho} = \Lambda & \qquad \rho \in [0, 1]\\
\rho^+ = 0 & \qquad \rho = 1
\end{matrix}
\right.
\end{equation}
with $\Lambda$ a positive scalar for tuning.
Clearly, the clock dynamics can be used to enforce jumps of the angular estimate at regular times and thus break the cylinder topological constraint, but it seems also convenient as a way to embed additional desirable features.
Among these, we will propose a simple identifier to enhance the observer transient performance.
Firstly, however, we introduce the baseline strategy with no identifier.
In place of \eqref{eq:reduced}, consider the hybrid system:
\begin{equation}\label{eq:reduced_hybrid}
\mathcal{H}_0: \left\{
\begin{matrix}
\begin{pmatrix}\dot{\eta}\\ \dot{\tilde{\xi}}\\ \dot{\rho} \end{pmatrix} = \begin{pmatrix}
\left( \chi\tilde{\xi} - k_{\eta} \chi \eta_2 \right)\mathcal{J}\eta\\
-\gamma \chi \eta_2\\
\Lambda
\end{pmatrix} \eqqcolon F_0(\eta, \tilde{\xi}, \rho, \chi) & \begin{pmatrix}\eta \\ \tilde{\xi}\\ \rho \end{pmatrix} \in C_{\text{s}}\\
\begin{pmatrix}{\eta}^+\\ {\tilde{\xi}}^+\\ {\rho}^+ \end{pmatrix} \in \begin{pmatrix}
\left\{
\begin{split}
-F\eta, & \quad \chi\eta_1 \leq 0\\
\eta, & \quad \chi\eta_1 \geq 0
\end{split}
\right.\\
\tilde{\xi}\\
0
\end{pmatrix} \eqqcolon G_0(\eta, \tilde{\xi}, \rho, \chi) & \begin{pmatrix}\eta \\ \tilde{\xi}\\ \rho \end{pmatrix} \in D_{\text{s}}\\
\end{matrix}
\right.
\end{equation}
where $F = \diagonal\{1, -1\}$, while $C_{\text{s}} = \mathbb{S}^1\times\mathbb{R}\times[0, 1]$ and $D_{\text{s}} = \mathbb{S}^1\times\mathbb{R}\times\{1\}$.
In this structure, the angle $\eta$ is always reset to a value satisfying $\eta_1 \geq 0$, thus ensuring that the set $\bar{x}_{\text{u}} \times [0, 1]$ is not an attractor compatible with the data of system \eqref{eq:reduced_hybrid}.
In fact, the next result confirms that the proposed hybrid strategy removes the unstable manifold $\mathcal{R}_{\mathcal{U}}$.
\begin{lemma}\label{lemma1}
The set $\mathcal{A}_0 \coloneqq \bar{x}_{\textup{s}} \times [0, 1] \subset \mathbb{S}^1 \times \mathbb{R}^2$ is a uniformly preasymptotically stable attractor for the hybrid system \eqref{eq:reduced_hybrid}, with basin of preattraction given by $\mathbb{S}^1\times\mathbb{R}^2$.
\end{lemma}
\begin{proof}
It is a direct application of the Nested Matrosov Theorem for hybrid systems \cite[Theorem 4.1]{sanfelice2009asymptotic}.
Indeed, consider the following Matrosov functions (which are continuous in their arguments, and thus bounded in any compact set of the states $(\eta, \tilde{\xi}, \rho)$, by Assumption \ref{hyp:w_regularity}):
\begin{equation}\label{eq:matrosov_fn}
\begin{split}
W_1(\eta, \tilde{\xi}, \rho, \chi) &= 1 - \eta_1 + \frac{1}{2\gamma}\tilde{\xi}^2\\
W_2(\eta, \tilde{\xi}, \rho, \chi) &= -\chi \tilde{\xi} \eta_1 \eta_2\\
W_3(\eta, \tilde{\xi}, \rho, \chi) &= \exp(\rho)\left[\eta_2^2 + \tilde{\xi}^2\right]\\
W_4(\eta, \tilde{\xi}, \rho, \chi) &= \exp(-\rho)\left[1 - \eta_1\right].
\end{split}
\end{equation}
Employing routine calculations and by means of by Assumption \ref{hyp:w_regularity}, it is possible to establish the bounds $\sup_{f \in F_0(\eta, \tilde{\xi}, \rho, \chi)}\langle \nabla W_i(\eta, \tilde{\xi}, \rho, \chi), (f, D^+\chi) \rangle \leq B_{\text{c}, i}(\eta, \tilde{\xi}, \rho)$, $i \in \{1, 2, 3, 4\}$, for all $(\eta, \tilde{\xi}, \rho) \in C_{\text{s}}$:
\begin{equation}
\begin{split}
B_{\text{c}, 1} &= -k_{\eta}\chi_{\text{m}}\eta_2^2 \leq 0\\
B_{\text{c}, 2} &= -\chi_{\text{m}}^2\eta_1^2\tilde{\xi}^2 + \Delta_2(M, \chi_{\text{M}}, \tilde{\xi}, \eta)|\eta_2|\\
B_{\text{c}, 3} &=  \Lambda W_2 + \exp(\rho)\Delta_3(\chi_{\text{M}}, \tilde{\xi}, \eta)|\eta_2|\\
B_{\text{c}, 4} &= -\Lambda \exp(-\rho)(1 - \eta_1) + \Delta_4(\chi_{\text{M}}, \tilde{\xi}, \eta)|\eta_2|,
\end{split}
\end{equation}
with $\Delta_2$, $\Delta_3$, $\Delta_4$ positive continuous functions in their arguments.
Note that $B_{\text{c}, 2} \leq -\chi_{\text{m}}^2\tilde{\xi}^2$ as $\eta_2 = 0$, thus in $B_{\text{c}, 3}$ and $B_{\text{c}, 4}$ the conditions 1)-2) of \cite[Theorem 4.1]{sanfelice2009asymptotic} must be checked in particular for $\eta_1 = -1$, $\eta_2 = 0$, $\tilde{\xi} = 0$, for any $\rho \in [0, 1]$.
Similarly, it holds $\sup_{g \in G_0(\eta, \tilde{\xi}, \rho, \chi)} W_i(g) - W_i(\eta, \tilde{\xi}, \rho, \chi) \leq B_{\text{d}, i}(\eta, \tilde{\xi}, \rho)$, $i \in \{1, 2, 3, 4\}$, with the following bounds, for all $(\eta, \tilde{\xi}, \rho) \in D_{\text{s}}$:
\begin{equation}
\begin{split}
B_{\text{d}, 1} &= \min \{0, 2\eta_1\} \leq 0\\
B_{\text{d}, 2} &= \max\{0, -2\chi_{\text{M}} |\tilde{\xi}| |\eta_2|\eta_1 \}, \qquad\qquad (B_{\text{d}, 2} > 0 \Rightarrow \eta_1 < 0)\\
B_{\text{d}, 3} &=  [\exp(0) - \exp(1)](\eta_2^2 + \tilde{\xi}^2) \leq 0\\
B_{\text{d}, 4} &= [\exp(0) - \exp(-1)] (1 - |\eta_1|).
\end{split}
\end{equation}
It can be easily verified from the first three bounds that the conditions 1)-2) of \cite[Theorem 4.1]{sanfelice2009asymptotic} are satisfied for all $(\eta, \tilde{\xi}, \rho) \in D_{\text{s}} \backslash \mathcal{A}_0$.
Finally, note that uniform global stability is easily established with $B_{\text{c}, 1}$, $B_{\text{d}, 1}$, in connection with the fact that $W_1$ is positive definite (considering a proper indicator function) with respect to the attractor $\mathcal{A}_0$, for all $(\eta, \tilde{\xi}, \rho) \in C_{\text{s}} \cup D_{\text{s}} \cup G_0(D_{\text{s}})$.
Since all sufficient conditions in \cite{sanfelice2009asymptotic} are verified, the statement follows immediately.
\end{proof}
In order to implement the hybrid observer leading to the above reduced order system, we need to compute the jumps of $\hat{\zeta}_\chi$ corresponding to $\eta^+ = -F\eta$, using $\hat{h}$ as a proxy of $h = -\chi\mathcal{J}\eta$.
For, note that
\begin{equation}
-F\eta = \mathcal{C}^T[\hat{\zeta}_\chi^+]\zeta_\chi = \mathcal{C}[\zeta_\chi]F\hat{\zeta}_\chi^+, 
\end{equation}
therefore it is possible to express $\hat{\zeta}_\chi^+$ as:
\begin{equation}
\hat{\zeta}_\chi^+ = -F\mathcal{C}^T[\zeta_\chi]F\eta = -\mathcal{C}[\zeta_\chi]\eta = -\mathcal{C}^T[\hat{\zeta}_\chi][\mathcal{C}[\zeta_\chi]\zeta_\chi].
\end{equation}
Furthermore, at each time a ``fast'' estimate of the rotor position (rescaled by $\chi > 0$) can be retrieved from $\hat{h}$ and $\hat{\zeta}_{\chi}$, since $\mathcal{J}h = \chi \eta$, and therefore $\chi \zeta_{\chi} = \mathcal{C}[\hat{\zeta}_{\chi}]\mathcal{J}h$.
These considerations finally yield the complete jump map $G_\zeta:\mathbb{R}^2\times\mathbb{S}^1 \rightrightarrows \mathbb{S}^1$:
\begin{equation}
\begin{split}
&G_\zeta(\hat{h}, \hat{\zeta}_\chi) \in \left\{\begin{matrix}
-\mathcal{C}^T[\hat{\zeta}_\chi]\begin{pmatrix}\cos(2\theta_\chi (\hat{h}, \hat{\zeta}_\chi))\\ \sin(2\theta_\chi (\hat{h}, \hat{\zeta}_\chi)) \end{pmatrix} &\quad \hat{h}_2 \geq 0\\
\hat{\zeta}_\chi &\quad \text{otherwise}
\end{matrix}
\right.\\
&\theta_\chi = \atan(y_\chi, x_\chi) \subset [-\pi, \pi],\qquad \begin{pmatrix} x_\chi \\ y_\chi \end{pmatrix} = \mathcal{C}[\hat{\zeta}_\chi]\mathcal{J}\hat{h}
\end{split}
\end{equation}
where in particular we let $\atan(0, 0) = [-\pi, \pi]$ and $\atan(y, x) = \{-\pi, \pi\}$, for all $(x, y)$ in the set $S = \{(x, y) \in \mathbb{R}^2: x < 0, y = 0\}$.
For convenience, let the map $G_{\text{f}}(\hat{h}, \hat{\zeta}_\chi) = \mathcal{C}^T[G_\zeta(\hat{h}, \hat{\zeta}_\chi)]\mathcal{C}[\hat{\zeta}_\chi]$ indicate the change of coordinates from the $\hat{\zeta}_\chi$-frame to the $\hat{\zeta}_\chi^+$-frame.
This map, which is available for observer design, is fundamental to describe the jumps that occur to both $i_{\hat{\chi}}$ and $h$, indeed:
\begin{equation}
i_{\hat{\chi}}^+ = \mathcal{C}^T[\hat{\zeta}_\chi^+]i_s = \mathcal{C}^T[G_\zeta(\hat{h}, \hat{\zeta}_\chi)]\mathcal{C}[\hat{\zeta}_\chi]\mathcal{C}^T[\hat{\zeta}_\chi]i_s = G_{\text{f}} i_{\hat{\chi}}, \qquad h^+ = -\chi\mathcal{J}\mathcal{C}^T[\hat{\zeta}_\chi^+]\zeta_\chi =  G_{\text{f}}h
\end{equation}
It follows that the overall observer structure is given by:
\begin{equation}\label{eq:hybrid_obs_v1}
\begin{matrix}
\begin{pmatrix}\dot{\hat{\imath}} \\ \dot{\hat{h}}\\ \dot{\hat{\zeta}}_\chi \\ \dot{\hat{\xi}} \\ \dot{\rho} \end{pmatrix} = \begin{pmatrix}
-\frac{R}{L}\hat{\imath} + \frac{1}{L}u_{\hat{\chi}} + \frac{\hat{h}}{L} - \hat{\omega}_\chi\mathcal{J}i_{\hat{\chi}} + k_{\text{p}}\tilde{\imath}\\
k_{\text{i}}\tilde{\imath}\\
\hat{\omega}_\chi\mathcal{J}\hat{\zeta}_\chi\\
\gamma \hat{h}_1\\
\Lambda
\end{pmatrix} & \rho \in [0, 1]\\
\begin{pmatrix}{\hat{\imath}}^+ \\ {\hat{h}}^+\\ {\hat{\zeta}}_\chi^+ \\ {\hat{\xi}}^+ \\ {\rho}^+ \end{pmatrix} \in \begin{pmatrix}
G_{\text{f}}(\hat{h}, \hat{\zeta}_\chi)\hat{\imath}\\
G_{\text{f}}(\hat{h}, \hat{\zeta}_\chi)\hat{h}\\
G_\zeta(\hat{h}, \hat{\zeta}_\chi)\\
\hat{\xi}\\
0
\end{pmatrix} & \rho = 1
\end{matrix}
\end{equation}
with $\hat{\omega}_\chi = |\hat{h}|\hat{\xi} + k_{\eta}\hat{h}_1$ as before.
Let $x_{\text{s}} \coloneqq (\eta, \tilde{\xi}, \rho) \in \mathbb{S}^1\times\mathbb{R}\times[0, 1]$ and $x_{\text{f}}\ \coloneqq T(\tilde{\imath}, \tilde{h}) \in \mathbb{R}^4$ with $\tilde{h} \coloneqq h - \hat{h}$ and $T$ a change of coordinates matrix such that \cite{tilli2019towards}:
\begin{equation}
x_{\text{f}} = T\begin{pmatrix}\tilde{\imath}\\ \tilde{h} \end{pmatrix} = \begin{pmatrix}
\varepsilon^{-1}I_2 & 0_{2\times 2}\\
-\varepsilon^{-1}I_2 & L^{-1}I_2
\end{pmatrix}\begin{pmatrix}\tilde{\imath}\\ \tilde{h} \end{pmatrix},
\end{equation}
with $\varepsilon$ a positive scalar such that $R/L + k_{\text{p}} = 2\varepsilon^{-1}$, $k_{\text{i}} = 2L\varepsilon^{-2}$.
We can then define the overall error dynamics as follows:
\begin{equation}\label{eq:error_system_hybrid}
\begin{matrix}
\begin{pmatrix}D^+x_{\text{f}}\\ \dot{x}_{\text{s}} \end{pmatrix} = \begin{pmatrix}
\varepsilon^{-1}\underbrace{\begin{pmatrix}-I_2 & I_2\\ -I_2& -I_2 \end{pmatrix}}_{A_\text{f}} x_{\text{f}} + \underbrace{\begin{pmatrix}0_{2\times2} \\ L^{-1}I_2 \end{pmatrix}}_{B_{\text{f}}} f_h\\
F_{\text{s}}(x_{\text{f}}, \chi, x_{\text{s}})
\end{pmatrix} & x_{\text{s}} \in C_{\text{s}}\\
\begin{pmatrix}x_{\text{f}}+\\ x_{\text{s}}^+ \end{pmatrix} \in \begin{pmatrix}
\diagonal\{G_{\text{f}}, G_{\text{f}}\}x_{\text{f}}\\
G_{\text{s}}(x_{\text{f}}, \chi, x_{\text{s}})
\end{pmatrix} & x_{\text{s}} \in D_{\text{s}}
\end{matrix}
\end{equation}
with $f_h = D^+h$ defined exactly as in \cite{tilli2019towards}, and $F_{\text{s}}$, $G_{\text{s}}$ the flows and jumps of the attitude estimation error (which correspond to the data in \eqref{eq:reduced_hybrid} if $\tilde{h} = 0$), respectively.
Note that it holds $A_{\text{f}} + A_{\text{f}}^T = -2I_4$, while the jump $x_{\text{f}}^+$ preserves the norm, indeed:
\begin{equation}
\begin{split}
|x_{\text{f}}^+|^2 & = |\varepsilon^{-1}\tilde{\imath}^+|^2 + |L^{-1}\tilde{h}^+ -\varepsilon^{-1}\tilde{\imath}^+|^2\\
& =  |G_{\text{f}} \varepsilon^{-1}\tilde{\imath}|^2 + |G_{\text{f}} (L^{-1}\tilde{h} -\varepsilon^{-1}\tilde{\imath})|^2\\
& = |\varepsilon^{-1}\tilde{\imath}|^2 + |L^{-1}\tilde{h} -\varepsilon^{-1}\tilde{\imath}|^2 = |x_{\text{f}}|^2.
\end{split}
\end{equation}
This means that on the one hand, during flows, the $x_{\text{f}}$-subsystem can be made arbitrarily fast by choosing $\varepsilon$ sufficiently small, while on the other hand the jumps do not cause any increase of $|x_{\text{f}}|$, and thus they do not represent an obstacle to time scale separation.
We summarize the stability properties of the above hybrid system with the following theorem, which represents the main result of this work.
\begin{theorem}\label{main_thm}
Consider system \eqref{eq:error_system_hybrid} with inputs $\chi(\cdot)$, $D^+\chi(\cdot)$, satisfying Assumption \ref{assumption}, and denote its solutions with $(\psi_{\textup{f}}(\cdot), \psi_{\textup{s}}(\cdot))$, with initial conditions $(x_{{\textup{f}}, 0}, x_{{\textup{s}}, 0})$.
In particular, denote with $\rho_0$ the initial condition of the clock.
Then, the attractor $0_{4\times1}\times\mathcal{A}_0$ is semiglobally practically asymptotically stable as $\varepsilon \to 0^+$, that is:
\begin{itemize}
\item there exists a proper indicator function $\sigma_{\textup{s}}$ of $\mathcal{A}_0$ in $\mathbb{S}^1\times \mathbb{R}^2$;
\item there exists a class $\mathcal{KL}$ function $\beta_{\textup{s}}$;
\end{itemize}
such that, for any positive scalars $\Delta_{\textup{f}}$, $\Delta_{\textup{s}}$, $\delta$, there exists a scalar $\varepsilon^* > 0$ such that, for all $0 < \varepsilon \leq \varepsilon^*$, all $(\psi_{\textup{f}}(\cdot), \psi_{\textup{s}}(\cdot))$ satifying $\rho_0 = 0$, $|x_{{\textup{f}}, 0}| \leq \Delta_{\textup{f}}$ and $\sigma_{\textup{s}}(x_{{\textup{s}}, 0}) \leq \Delta_{\textup{s}}$, the following bounds hold, for all $(t, j) \in \dom(\psi_{\textup{f}}(\cdot), \psi_{\textup{s}}(\cdot))$:
\begin{equation}
\begin{split}
|\psi_{\textup{f}}(t, j)| &\leq \exp\left(-t/\varepsilon \right)|x_{{\textup{f}}, 0}| + \delta\\
\sigma_{\textup{s}}(\psi_{\textup{s}}(t, j)) &\leq \beta_{\textup{s}}(\sigma_{\textup{s}}(x_{{\textup{s}}, 0}), t + j) + \delta.
\end{split}
\end{equation}
\end{theorem}
Figure \ref{fig:sim_results} (plots (b),(e),(h),(k)) presents the simulation results corresponding to observer \eqref{eq:hybrid_obs_v1}, with $\Lambda$ selected as in Table \ref{tab:TabParMot}. Notably, it is possible to appreciate that this new solution enhances the convergence performance of the previous continuous-time algorithm. This is motivated by the intuition that the jumps, for $\Lambda$ sufficiently large, impose the position estimation error to be close during transients either to $\eta = (0, \;1)$ or to $\eta = (0, \; -1)$: these configurations are associated with the maximal value of $\dot{\hat{\xi}}$. For this reason, we can expect that there exists a range for initial conditions of $|\tilde{\xi}|$ where the convergence properties of this observer are optimized. In particular, this range is expected to be between very large initial errors, where the continuous time angular ``wraps'' dominate the behavior, and small initial errors, where jumps do not cause any estimation correction.
\subsection{A Mini-Batch Identifier for Enhanced Initial Convergence}
We conclude this section with a modification of the above strategy to ensure a faster observer response, obtained by means of a discrete-time identifier.
The need to employ a higher number of state variables, in addition to performing the minimization of a cost function, clearly makes this method more computationally intensive.
However, some strategies can be adopted to mitigate the online burden and enable implementation in embedded computing systems (e.g., moving the procedure in a lower priority/frequency task).\\
Firstly, recall that a perturbed estimate of $\chi \zeta_{\chi}$ can be computed as $\mathcal{C}[\hat{\zeta}_{\chi}]\mathcal{J}\hat{h}$.
From the solutions of system \eqref{eq:chi_frame} it can be noted that, for any positive scalar $T$, for all $t \geq T$:
\begin{equation}
\zeta_{\chi}(t) - \zeta_{\chi}(t-T) = \xi \mathcal{J} \int_{t - T}^{t}\chi(s)\zeta_{\chi}(s)ds.
\end{equation}
Hence by multiplying both sides by $\chi(t - T)\chi(t)$ it follows (let $y(s) = \chi(s)\zeta_{\chi}(s)$):
\begin{equation}\label{eq:lin_regr_pmsm}
\chi(t - T)y(t)- \chi(t)y(t - T) = \xi\chi(t - T)\chi(t)\mathcal{J} \int_{t - T}^{t}y(s)ds,
\end{equation}
which can be constructed by means of division-free estimates, since $\chi$ can be replaced with $|\hat{h}|$ and $y$ with $\mathcal{C}[\hat{\zeta}_{\chi}]\mathcal{J}\hat{h}$.
Indeed, between jumps of the clock \eqref{eq:timer}, we can compactly rewrite \eqref{eq:lin_regr_pmsm} as $X(t, j) + e_X(t, j) = (\Phi(t, j) + e_\Phi(t, j))\xi$, where $X$ and $\Phi$ are only function of $\hat{h}$, $\hat{\zeta}_\chi$, their past values and their integrals, while $e_X$ and $e_\Phi$ are disturbances depending on $h$ and $\tilde{h}$.
For $N \in \mathbb{N}_{\geq 1}$, let $\tau_N(\cdot)$ be a moving window operator such that, for a hybrid arc $\psi$ satisfying jumps according to the clock \eqref{eq:timer} (with $\rho(0, 0) = 0$), and for all $(t, j) \in \dom \psi$ such that $j \geq N$:
\begin{equation}
\tau(\psi)(t, j) = \begin{pmatrix}
\psi\left((j - N + 1)/\Lambda, j- N\right)\\
\vdots\\
\psi\left(j/\Lambda, j-1\right)
\end{pmatrix}.
\end{equation}
Choosing $T = 1/\Lambda$ as interval of integration in \eqref{eq:lin_regr_pmsm}, we thus obtain a simple estimate of $\xi$ through a batch least-squares algorithm as follows (see \cite{bin2019approximate} for the same structure in the context of output regulation):
\begin{equation}
\begin{split}
\xi^*(t, j) & =  \argmin_{\theta \in \mathbb{R}} J_N(\theta)(t, j)\\
J_N(\theta)(t, j) &\coloneqq \left| \tau_N(X)(t, j) - \tau_N(\Phi)(t, j)\theta \right|^2.
\end{split}
\end{equation}
To implement the above strategy, the hybrid observer in \eqref{eq:hybrid_obs_v1} is augmented with an identifier based on the shift register variables $Y^\mu = (Y_{0}^\mu, \ldots, Y_{N}^\mu) \in \mathbb{R}^{2(N+1)}$, $Z^\mu = (Z_{0}^\mu, \ldots, Z_{N}^\mu) \in \mathbb{R}^{N + 1}$, $\Phi^\mu = (\Phi_{1}^\mu, \ldots, \Phi_{N}^\mu) \in \mathbb{R}^{2N}$, related to the moving window operator as $\tau_N(\Phi) = \Phi^\mu$, $\tau_N(X) = (X_{1}^\mu, \ldots, X_{N}^\mu)$, $X_{i}^\mu = Z_{i-1}^\mu Y_{i}^\mu - Z_{i}^\mu Y_{i-1}^\mu$, $i \in \{1, \ldots, N\}$):
\begin{equation}\label{eq:batch}
\begin{split}
&\left\{\begin{split}
\dot{\nu} &= \mathcal{C}[\hat{\zeta}_\chi]\mathcal{J}\hat{h}\\
\dot{Y}^\mu &= 0\\
\dot{Z}^\mu &= 0\\
\dot{\Phi}^\mu &= 0
\end{split}\right. \qquad \qquad \qquad \qquad \qquad \;\; \rho \in [0, 1]\\
&\left\{\begin{split}
\nu^+ &= 0\\
(Y^\mu_{i})^+ &= Y_{i+1}^\mu, \qquad i \in \{ 0, \ldots, N \}\\
(Y^\mu_{N})^+ &= \mathcal{C}[\hat{\zeta}_{\chi}]\mathcal{J}\hat{h}\\
(Z^\mu_{i})^+ &= Z_{i+1}^\mu, \qquad i \in \{ 0, \ldots, N \}\\
(Z^\mu_{N})^+ &= |\hat{h}|\\
(\Phi^\mu_{i})^+ &= \Phi_{i+1}^\mu, \qquad i \in \{ 1, \ldots, N \}\\
(\Phi^\mu_{N})^+ &= \mathcal{J} \nu |\hat{h}|Z_N^\mu
\end{split}\right. \qquad \rho = 1\\
&\xi^*(t, j) = \mathcal{G}[Y^\mu, Z^\mu, \Phi^\mu](t, j) = \argmin_{\theta \in \mathbb{R}} J_N(\theta)(t, j),
\end{split}
\end{equation}
where the standard Moore-Penrose pseudoinverse can be used to minimize $J_N$.
The jump map of $\hat{\xi}$ can be then modified as a function of $\hat{\xi}(t, j)$ and $\xi^*(t, j)$.
Without intending to provide a formal stability result for this modification, which will be the topic of future research activity, we propose to jump according to two criteria, which are the ``readiness'' of the shift register and the error $\hat{\xi} - \xi^*$:
\begin{equation}\label{eq:batch_jump}
\hat{\xi}^+ = \left\{
\begin{split}
&\hat{\xi} \qquad j \leq N + 1 \text{ or }  |\hat{\xi} - \xi^*| \leq 4\sqrt{\gamma}\\
&\xi^* \qquad \text{ otherwise}.
\end{split}
\right.
\end{equation}
This way it is possible to ensure that, if the regression errors $e_X$, $e_\Phi$ are sufficiently small, the above jump improves the estimate $\hat{\xi}$ by guaranteeing $x_{\text{s}}^+$ to be close to the set $W_1 \leq 2$ in \eqref{eq:matrosov_fn} (where $4\sqrt{\gamma}$ was employed to account for the worst case scenario).
Within such set, the local behavior of the attitude observer becomes dominant, guaranteeing a desirable residual behavior.
Note that the errors $e_X$, $e_\Phi$ can be made arbitrarily small, for any jump of the overall system.
This is possible because $e_X$, $e_\Phi$ vanish as $\tilde{h} \to 0$ and, by proper selection of the gains of the fast subsystem, $\tilde{h}$ can be forced to converge during flows in an arbitrarily small ball, before the first jump occurs.\\
Finally, Figure \ref{fig:sim_results} (plots (c),(f),(i),(l)) presents the simulation results corresponding to the augmented observer \eqref{eq:hybrid_obs_v1}-\eqref{eq:batch}-\eqref{eq:batch_jump}, with $N$ chosen as in Table \ref{tab:TabParMot}. As expected, the proposed discrete-time identifier further improves on the previous solution in terms of estimation speed. Indeed, the fast reduction of $\tilde{\xi}$ is obtained after a brief waiting time according to \eqref{eq:batch_jump}.

\section{Conclusions}\label{sec:conclusions}

We presented a hybrid sensorless observer for PMSMs, with no a priori knowledge of the mechanical model.
The rotor speed was assumed to be an unknown input disturbance (endowed with some mild regularity assumptions) with constant, unknown sign and a persistently non-zero magnitude.
We showed that a clock, used to trigger an appropriate position estimation update, is sufficient to yield a semiglobal practical stability result in this challenging scenario.
Motivated by the improved convergence properties with respect to an initial continuous-time solution, we also proposed an estimation speed-up strategy based on a discrete-time identifier.
Future research effort will be dedicated to relaxing the required speed assumptions, as well as to further developments of the presented discrete-time identification technique.

\bibliographystyle{ieeetr}
{\bibliography{ifac2020_bibliography}}\
                                                   
\end{document}